\documentclass[final,3p,times,twocolumn]{elsarticle}
\usepackage{amssymb}
\usepackage{amsmath}
\usepackage{url}
\usepackage{graphicx}
\usepackage{subfig}
\usepackage{graphicx}
\usepackage{float}
\usepackage{rotating}
\usepackage{orcidlink}
\usepackage{tikz}
\usepackage{hyperref}
\usepackage{enumitem}
\usepackage{stfloats}
\usepackage{amsthm}
\usepackage{subcaption}

\newtheorem{theorem}{Theorem}

\newtheorem{lemma}{Lemma}
\newtheorem{definition}{Definition}
\newtheorem{conjecture}{Conjecture}
\definecolor{orcidlogocol}{HTML}{A6CE39}

\newcommand{\orcidicon}{%
    \tikz[baseline=-0.5ex]\node[shape=circle,fill=orcidlogocol,inner sep=1pt] {\tiny\textsf{ID}};%
}

\newcommand{\orcid}[1]{\href{https://orcid.org/#1}{\orcidicon}}

\usepackage{booktabs}

\begin{document}

\begin{frontmatter}

%% Title, authors and addresses

%% use the tnoteref command within \title for footnotes;
%% use the tnotetext command for theassociated footnote;
%% use the fnref command within \author or \affiliation for footnotes;
%% use the fntext command for theassociated footnote;
%% use the corref command within \author for corresponding author footnotes;
%% use the cortext command for theassociated footnote;
%% use the ead command for the email address,
%% and the form \ead[url] for the home page:
%% \title{Title\tnoteref{label1}}
%% \tnotetext[label1]{}
%% \author{Name\corref{cor1}\fnref{label2}}
%% \ead{email address}
%% \ead[url]{home page}
%% \fntext[label2]{}
%% \cortext[cor1]{}
%% \affiliation{organization={},
%%             addressline={},
%%             city={},
%%             postcode={},
%%             state={},
%%             country={}}
%% \fntext[label3]{}

\title{Correlation-Weighted Communicability Curvature as a Structural Driver of Dengue Spread: A Bayesian Spatial Analysis of Recife (2015–2024)}

%% use optional labels to link authors explicitly to addresses:
%% \author[label1,label2]{}
%% \affiliation[label1]{organization={},
%%             addressline={},
%%             city={},
%%             postcode={},
%%             state={},
%%             country={}}
%%
%% \affiliation[label2]{organization={},
%%             addressline={},
%%             city={},
%%             postcode={},
%%             state={},
%%             country={}}

\author[1]{Marcílio Ferreira dos Santos}
\ead{marcilio.santos@ufpe.br}
\ead[url]{https://orcid.org/0000-0001-8724-0899}

\author[2]{Cleiton de Lima Ricardo}
\ead{cleiton.lricardo@ufpe.br}
\ead[url]{https://orcid.org/0000-0002-7114-1201}

\author[3]{Andreza dos Santos Rodrigues de Melo}
\ead{andreza.rodrigues@ufpe.br}
\ead[url]{https://orcid.org/0000-0002-5153-6548}

%\cortext[cor1]{Corresponding author.}

%% Optional: equal contribution note
%\fntext[eq1]{These authors contributed equally to this work.}

%% Affiliations
\affiliation[1]{
    organization={Núcleo de Formação de Docentes, Universidade Federal de Pernambuco (UFPE)}, 
    addressline={}, 
    city={Caruaru},
    postcode={},
    state={PE},
    country={Brazil}
}

\affiliation[2]{
    organization={Núcleo Interdisciplinar de Ciências Exatas e da Natureza (NICEN), Universidade Federal de Pernambuco (UFPE)}, 
    addressline={}, 
    city={Caruaru},
    postcode={},
    state={PE},
    country={Brazil}
}

\affiliation[3]{
    organization={Departamento de Engenharia Cartográfica, Universidade Federal de Pernambuco (UFPE)}, 
    addressline={}, 
    city={Recife},
    postcode={},
    state={PE},
    country={Brazil}
}

\begin{abstract}
We investigate whether functional connectivity in urban road networks explains dengue incidence in Recife, Brazil (2015--2024), beyond traditional adjacency-based spatial dependence. For each neighborhood, we compute the average \emph{communicability curvature}, a graph-theoretic measure capturing multiscale accessibility through redundant network paths. The curvature metric is incorporated into Negative Binomial models, fixed-effects regressions, SAR/SAC spatial models, and a hierarchical INLA/BYM2 specification. Across all frameworks, curvature emerges as the strongest and most stable predictor of dengue risk. In the BYM2 model, the structured spatial component collapses ($\phi \approx 0$), indicating that spatial variation traditionally attributed to
CAR adjacency effects is largely absorbed by functional network connectivity. Rather than eliminating spatial dependence, the results suggest a reparametrization of space: dengue diffusion in Recife is structured less by geometric contiguity and more by network-mediated urban connectivity.
\end{abstract}

%%Graphical abstract
\begin{graphicalabstract}
\end{graphicalabstract}

%%Research highlights
%%Research highlights
\begin{highlights}
\item Communicability curvature captures multiscale functional connectivity in urban dengue networks
\item Curvature outperforms climatic, demographic, and classical spatial predictors
\item Adjacency-based spatial effects are reparameterized by functional network structure
\item The curvature effect is robust across spatial cutoffs and model specifications
\item Out-of-sample forecasting confirms genuine predictive, non-circular information
\end{highlights}

%% Keywords
\begin{keyword}
Dengue, Spatial Epidemiology, Correlation-Weighted Communicability Curvature, Network Science, Bayesian Hierarchical Models, INLA
\end{keyword}

\end{frontmatter}

%% Add \usepackage{lineno} before \begin{document} and uncomment 
%% following line to enable line numbers
%% \linenumbers

%% main text
%%

%% Use \section commands to start a section

\section{Introduction}\label{sec:introduction}

Dengue remains one of the most significant vector-borne diseases worldwide, with
estimates indicating nearly 390 million infections annually and recurrent
outbreaks across tropical and subtropical regions. Its primary vector, the
\textit{Aedes aegypti} mosquito, exhibits complex ecological and epidemiological
dynamics that continue to challenge surveillance and control programs
\cite{barcellos2000, teixeira2009, honorio2009}. Despite extensive knowledge on
vector biology and environmental drivers, substantial uncertainty persists
regarding how human cases organize in space and time within urban environments,
and how such patterns can be translated into actionable indicators of epidemic
risk.

Classical epidemiological models, including compartmental frameworks such as SIR
and its extensions, are effective at capturing aggregate dynamics but often fail
to represent fine-scale spatial heterogeneity and localized transmission
processes \cite{keeling2008, carvalho2021, parselia2019}. In recent years,
graph-based and network-oriented approaches have emerged as a complementary
paradigm, allowing epidemic spread to be interpreted through patterns of
connectivity, clustering, and diffusion on complex structures
\cite{estrada2012, massaro2019}. However, the application of such methods to dengue
in dense urban settings remains limited, particularly with respect to integrating
functional connectivity beyond simple geographic adjacency.

The construction of biologically meaningful networks for dengue requires careful
consideration of spatial and temporal constraints. Empirical studies indicate
that the typical flight range of \textit{Aedes aegypti} lies between 100 and
200~m \cite{Moore2022}, while the extrinsic incubation period of the virus within
the mosquito ranges from approximately 6 to 14 days \cite{Pruszynski2024}. These
constraints support network formulations that capture either direct transmission
chains or co-infection clusters, in which multiple cases share common temporal
dynamics within restricted spatial windows \cite{simoy2015, zheng2019}.

In this study, we leverage high-resolution, street-level dengue surveillance data
from Recife, Brazil, where the geographic location of each reported case is
available. This enables the comparison of temporal incidence patterns across
urban units and the identification of functional relationships emerging from
their time series. To characterize the structural role of such relationships, we
employ communicability curvature \cite{estrada2012, estrada2012book}, a
graph-theoretic measure that quantifies the contribution of network elements to
multiscale connectivity and diffusion.

Rather than modeling individual transmission events, communicability curvature
acts as a structural descriptor of urban accessibility, summarizing how local
incidence patterns are embedded within the broader functional network. By
integrating this measure into statistical and spatial epidemiological models, we
aim to assess whether network-mediated connectivity provides explanatory and
predictive power beyond classical climatic, demographic, and adjacency-based
spatial effects.

\section{Graph-Based Modelling of Functional Transmission Topology}

In this section, we formalize the graph-based modelling framework adopted to
represent the spatiotemporal organization of dengue incidence. Graphs offer a
versatile representation because they encode both the geographic distribution
of cases and the temporal relationships underlying epidemic coevolution.
The construction follows established principles in network epidemiology and
serves as the structural basis for all analyses presented in this article.

\subsection{Local Time Series Model}

This model is closely aligned with formulations used in spatial epidemiology
and the study of complex systems
\cite{anselin1988, anselin1995, florax2003, lindgren2011, Massaro2019}.
Each vertex corresponds to a georeferenced street-level location associated
with a local dengue incidence time series, rather than individual cases.
Edges are established between vertices whose temporal similarity—here measured
by the Pearson correlation between their respective time series—exceeds a
predefined threshold, indicating epidemiological synchrony between distinct
urban locations.

Correlation is employed as a statistical proxy for temporal coevolution,
rather than as a direct representation of causal transmission.
Highly correlated series may reflect shared exposure to local vector
populations or indirect diffusion processes mediated by human mobility.

To illustrate the resulting structural connectivity,
Figure~\ref{fig:street_graph_overlay} shows the functional street-level graph
superimposed on the geographic map of Recife. This visualization highlights
the major channels of structural interaction induced by correlated dengue
activity and provides a geometric interpretation of the network structure
used throughout the study.

The epidemiological interpretation is twofold:
\begin{itemize}[label=--]
    \item[(i)] exposure to a common or overlapping source of infectious mosquitoes; or
    \item[(ii)] localized diffusion processes operating over time through
    spatially proximate areas.
\end{itemize}

Analogous correlation-based graph constructions have been widely explored in
recent epidemic studies, particularly during the COVID-19 pandemic, where
they were used to characterize systemic fragility, synchrony, and epidemic
risk \cite{carvalho2021, de2021using, chinazzi2020}.

Formally, this construction yields a graph $G = (V, A)$, where the vertex set
$V$ represents street-level locations and the edge set $A$ connects pairs of
vertices exhibiting statistically significant temporal synchrony.
Once constructed, the network topology is held fixed and subsequently analyzed
through structural and statistical models, ensuring that it functions as an
exogenous descriptor of urban connectivity.

To ensure biological plausibility, an additional spatial pruning criterion is
applied: only pairs of vertices separated by less than 600~m are connected.
This threshold corresponds to approximately three times the typical flight
range of \emph{Aedes aegypti} \cite{Moore2022}, accounting for geocoding
uncertainty and for the approximation inherent in representing streets by
centroids \cite{parselia2019, adeola2017}.

Among the candidate formulations evaluated, we focus on the local time series
model due to its conceptual alignment with established epidemiological
methodologies, its ability to capture fine-scale spatiotemporal synchrony,
and its direct applicability to routine municipal surveillance data.

\begin{figure*}[t]
    \centering
    \includegraphics[width=0.4\textwidth]{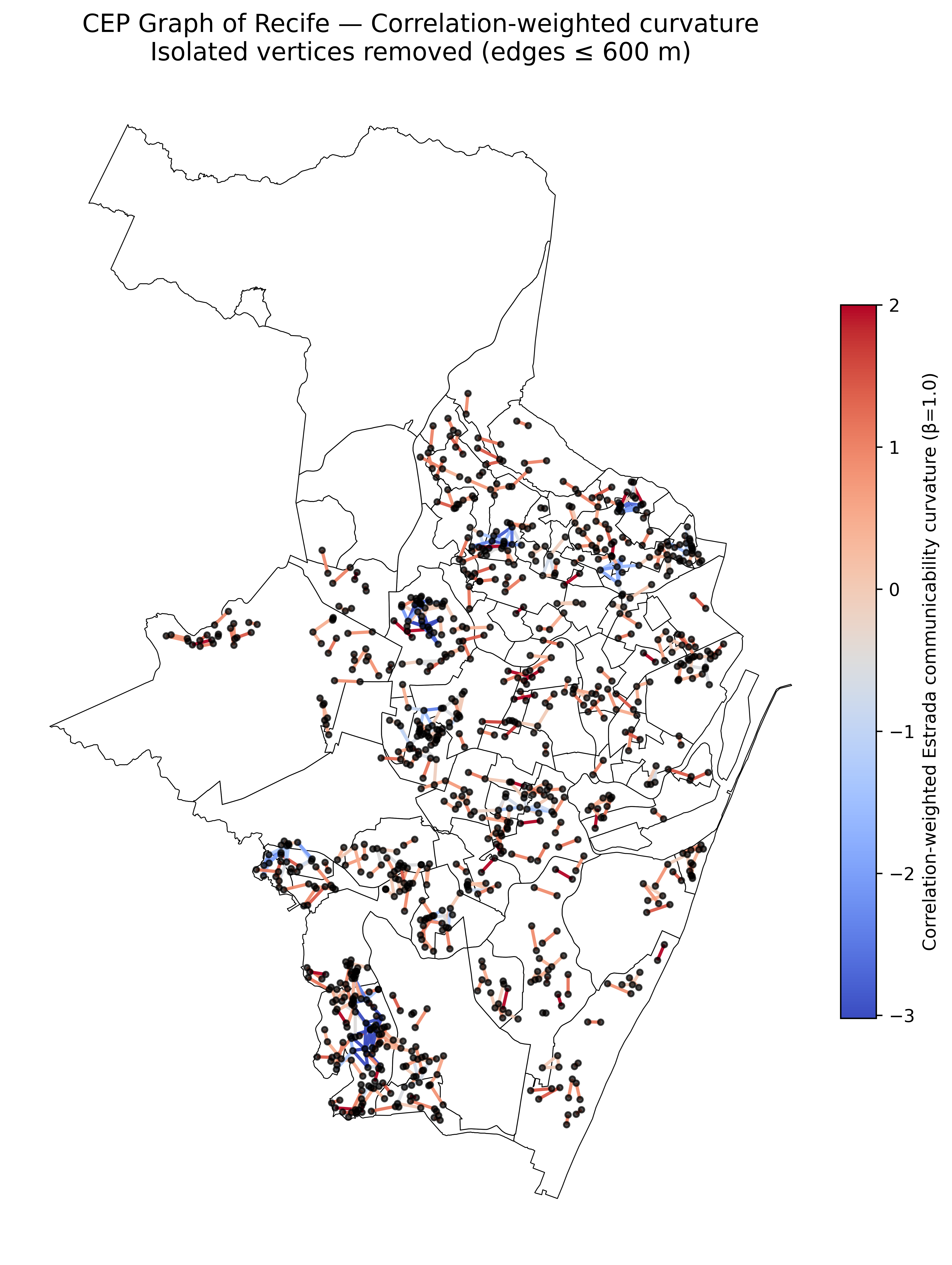}
    \caption{Street-level functional connectivity graph for Recife.
    Nodes represent street-level locations associated with local dengue
    incidence time series, and edges denote statistically significant
    temporal synchrony between series.
    The graph is superimposed on the geographic map to reveal the functional
    transmission topology and the spatial constraints shaping diffusive
    epidemic spread.}
    \label{fig:street_graph_overlay}
\end{figure*}

\section{Theoretical and Empirical Results}

In this section, we present the mathematical foundations supporting the use of
communicability curvature as a diffusion-sensitive descriptor in epidemiological
graphs, together with an empirical analysis based on dengue data from the
Metropolitan Region of Recife. The aim is to demonstrate, in an integrated
manner, both the formal consistency of the approach and its practical relevance
for monitoring and interpreting spatial patterns of transmission.

The analysis of mean curvature across neighborhoods reveals distinct spatial
configurations in the structure of dengue dissemination over the historical
series. Here, curvature is interpreted as an indicator of the structural
coherence of the functional transmission graph: more negative values correspond
to highly integrated network configurations that facilitate diffusion, whereas
higher values are associated with increased fragmentation and reduced
connectivity among nodes.

In 2015, we observe a predominance of negative or near-zero curvature across much
of the territory. This pattern is characteristic of a highly connected network
in which multiple potential epidemiological pathways coexist among different
regions of the city. Such a configuration is consistent with the intense and
spatially widespread outbreaks observed during that year. From a structural
perspective, the virus encountered few effective territorial barriers, forming
a functionally integrated urban mesh conducive to broad diffusion.

In 2019, the spatial pattern changes substantially: curvature values increase,
indicating a reduction in structural connectivity. This behavior aligns with the
marked decline in dengue cases recorded during that period. Higher curvature
suggests that transmission became more localized, with contagion dynamics
restricted to smaller subsets of neighborhoods. Structurally, 2019 corresponds
to a scenario of increased segmentation and greater resistance to large-scale
spatial propagation.

In 2024, we observe a partial reorganization of epidemiological connectivity.
Although the number of cases rises again, the spatial configuration does not
reproduce the highly integrated pattern observed in 2015. Curvature values
display greater heterogeneity, combining regions of positive curvature with
localized pockets of negative curvature. This configuration indicates a moderate
diffusion potential, concentrated in specific urban clusters. The network thus
recovers part of its connectivity, while remaining less integrated than during
the peak epidemic year.

A comparison across these three years reveals a consistent structural pattern:
(i) in 2015, the urban network exhibited high functional connectivity, favoring
widespread diffusion; (ii) in 2019, increased fragmentation substantially
limited transmission potential; and (iii) in 2024, the system assumed an
intermediate configuration characterized by localized outbreaks and moderate
connectivity. These results indicate that spatial curvature captures meaningful
changes in the structural organization of epidemiological risk, distinguishing
periods of heightened susceptibility to broad dissemination from those in which
transmission tends to remain spatially confined.

Thus, the maps presented do not merely reflect variations in case intensity;
they reveal how the spatial architecture of functional connectivity reorganizes
over time. Communicability curvature therefore emerges as a complementary tool
for geographically targeted epidemiological surveillance, supporting the
identification of structurally vulnerable areas and the design of more precise
intervention strategies.

% ==========================================================
% Figura: Curvatura 2015
% ==========================================================
\begin{figure}[ht]
    \centering
    \includegraphics[width=\linewidth]{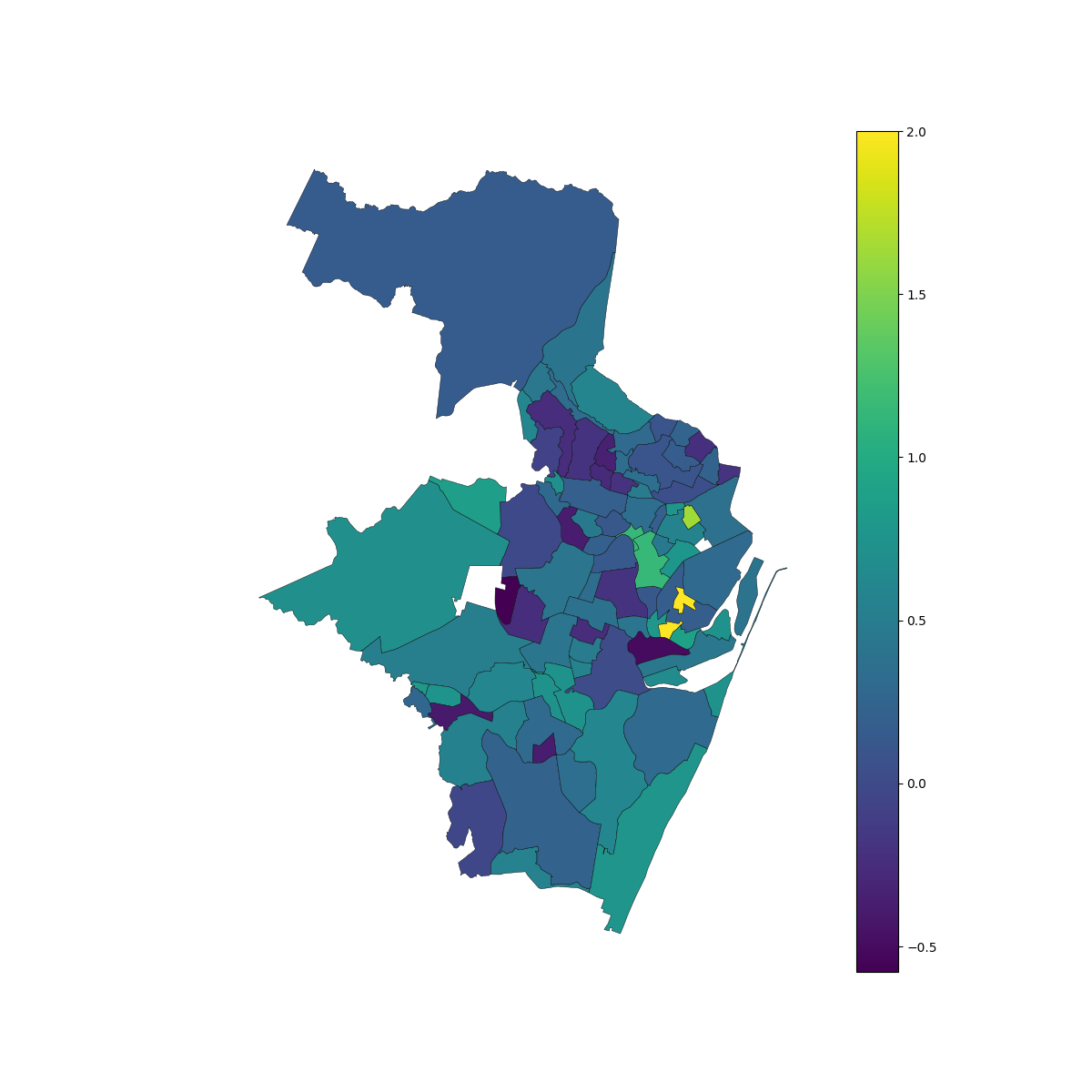}
    \caption{Mean communicability curvature by neighborhood in 2015.}
    \label{fig:curv2015}
\end{figure}

% ==========================================================
% Figura: Curvatura 2019
% ==========================================================
\begin{figure}[ht]
    \centering
    \includegraphics[width=\linewidth]{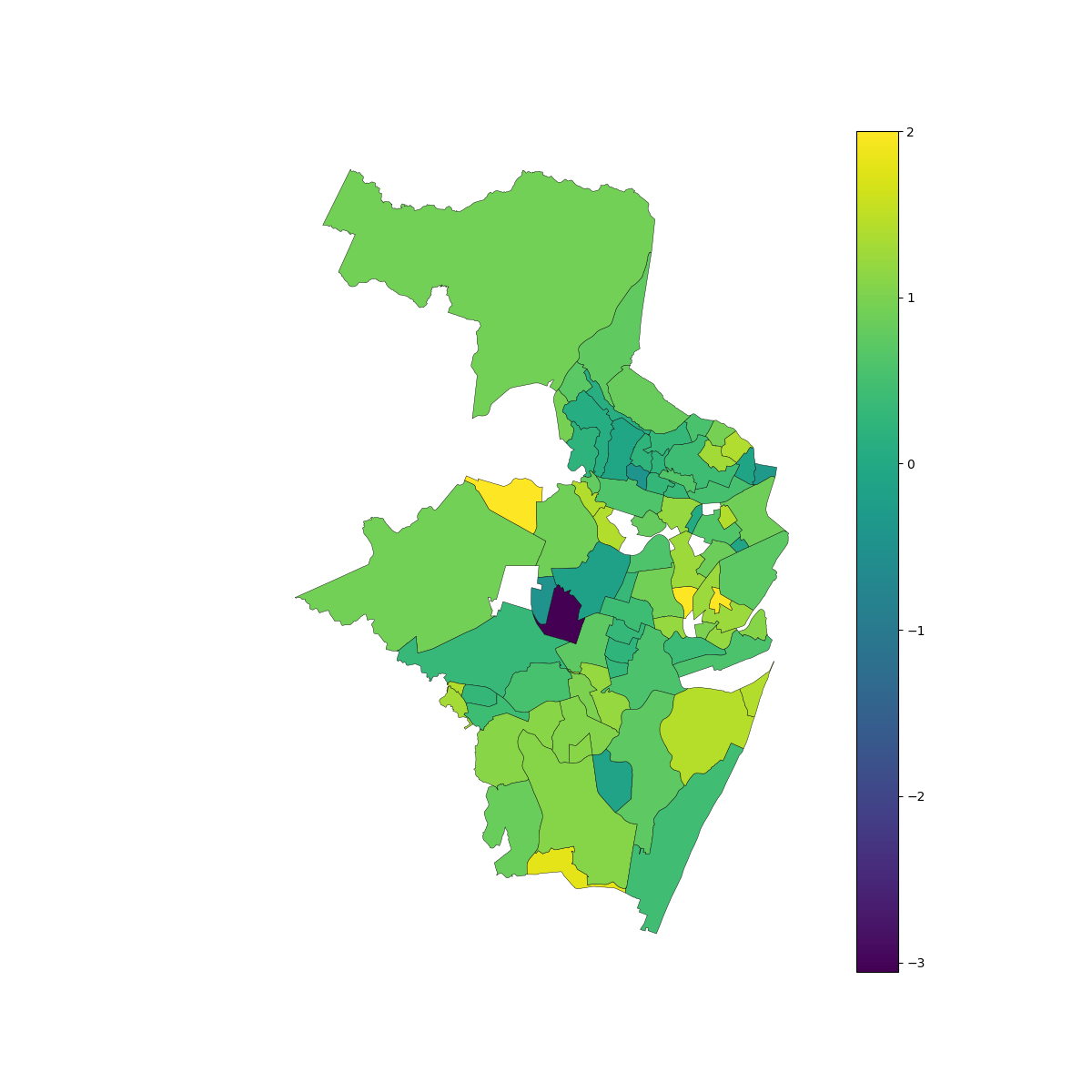}
    \caption{Mean communicability curvature by neighborhood in 2019.}
    \label{fig:curv2019}
\end{figure}

% ==========================================================
% Figura: Curvatura 2024
% ==========================================================
\begin{figure}[ht]
    \centering
    \includegraphics[width=\linewidth]{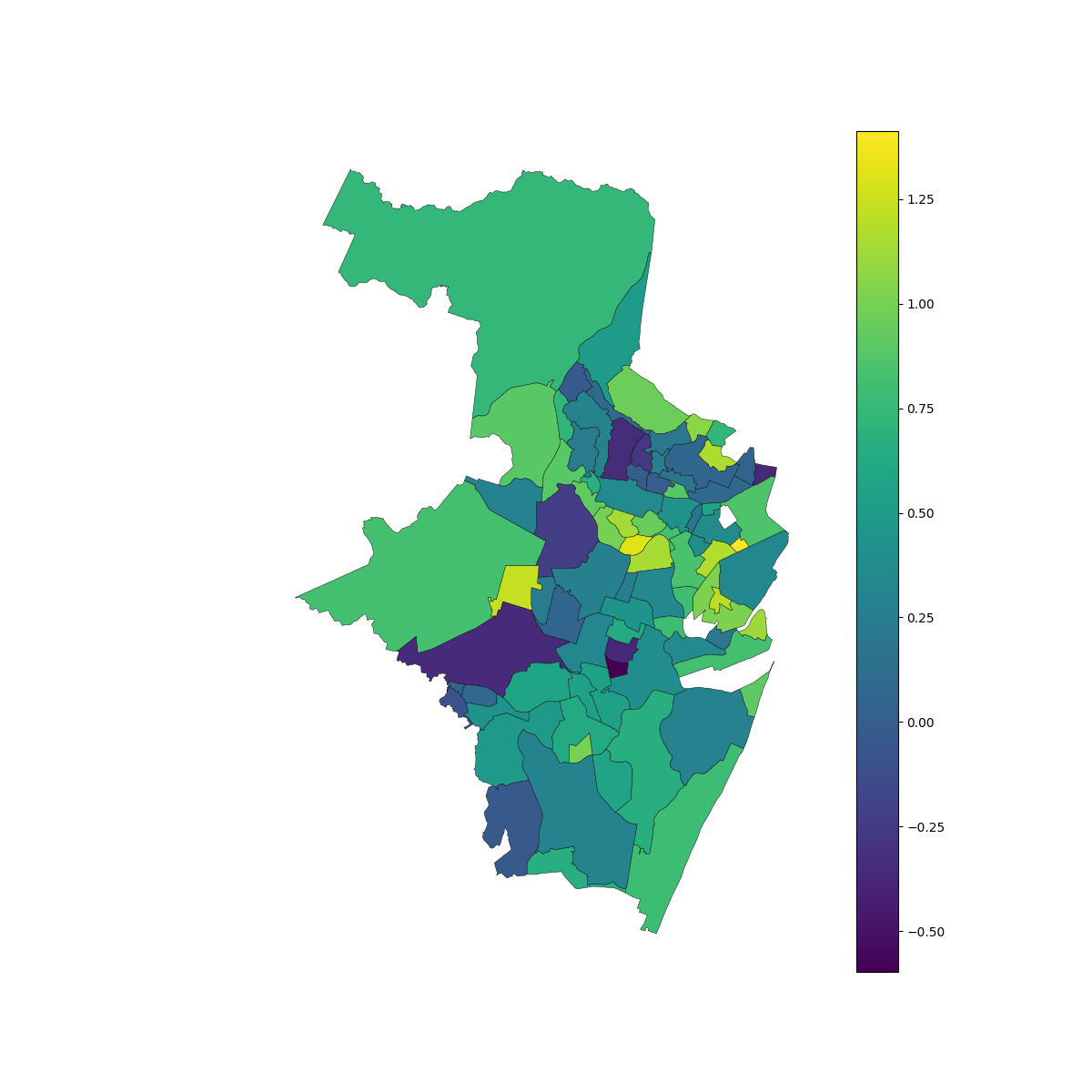}
    \caption{Mean communicability curvature by neighborhood in 2024.}
    \label{fig:curv2024}
\end{figure}

% --------------------------------------------------
\subsection{Communicability and Multiscale Accessibility}

Let $G=(V,E)$ be a simple, undirected graph with nonnegative adjacency matrix
$A=[A_{ij}]$. The \emph{communicability} between vertices $i$ and $j$ is defined as
\begin{equation}
C_{ij}(\beta) = \bigl(e^{\beta A}\bigr)_{ij},
\end{equation}
where the matrix exponential admits the expansion
\begin{equation}
e^{\beta A} = \sum_{k=0}^{\infty} \frac{\beta^k}{k!} A^k.
\end{equation}

Since $(A^k)_{ij}$ counts the number of walks of length $k$ between vertices $i$
and $j$, communicability aggregates all possible transmission routes, weighting
longer paths by exponentially decaying factors. The parameter $\beta>0$ controls
the scale of interaction: small values emphasize local structure, while larger
values incorporate increasingly global pathways.

The series converges absolutely for all $\beta>0$, since
\[
\lim_{k\to\infty} \frac{\beta^k}{k!} = 0,
\]
ensuring that very long walks contribute negligibly.

% --------------------------------------------------
\subsubsection{Diffusive Interpretation and Monotonicity}

\begin{lemma}[Monotonicity in the Diffusion Parameter]
\label{lem:monotonicity}
For any vertices $i,j\in V$ and parameters $0<\beta_1<\beta_2$,
\[
C_{ij}(\beta_1) \le C_{ij}(\beta_2).
\]
\end{lemma}

\begin{proof}
From the series expansion,
\[
C_{ij}(\beta) = \sum_{k=0}^{\infty} \frac{\beta^k}{k!} (A^k)_{ij},
\]
all coefficients and summands are nonnegative and increase monotonically in
$\beta$.
\end{proof}

The monotonicity of $C_{ij}(\beta)$ reflects the intensification of diffusion:
increasing $\beta$ enhances the contribution of longer transmission routes,
amplifying the effective connectivity of the network.

\begin{theorem}[Communicability as Linear Diffusion]
\label{thm:linear-diffusion}
Consider the linear system
\[
\frac{dX(t)}{dt} = \beta A X(t), \qquad X(0)=e_i.
\]
Its solution is
\[
X(t) = e^{\beta t A} e_i,
\]
and the influence of vertex $i$ on vertex $j$ at time $t$ is proportional to
$C_{ij}(\beta t)$.
\end{theorem}

\begin{proof}
This follows directly from the standard solution of linear differential equations
via the matrix exponential.
\end{proof}

This result establishes communicability as a natural approximation of linearized
diffusive processes on networks, including epidemic propagation in early or
locally linear regimes.

% --------------------------------------------------
\subsubsection{Mean Communicability and Spectral Structure}

Let $G_\tau$ denote the graph obtained by thresholding correlations at level
$\tau$, with adjacency matrix $A_\tau$ and eigenvalues
$\lambda_1(\tau)\ge\cdots\ge\lambda_n(\tau)$.  
The \emph{mean communicability} is defined as
\begin{equation}
\bar{C}(\tau)
= \frac{1}{|V|^2} \sum_{i,j\in V} \bigl(e^{\beta A_\tau}\bigr)_{ij}.
\end{equation}

\begin{lemma}[Spectral Representation of Mean Communicability]
\label{lem:spectral}
\[
\bar{C}(\tau)
= \frac{1}{|V|^2} \sum_{k=1}^n e^{\beta \lambda_k(\tau)}.
\]
\end{lemma}

\begin{proof}
By the spectral decomposition $A_\tau = Q_\tau \Lambda_\tau Q_\tau^\top$,
\[
e^{\beta A_\tau} = Q_\tau e^{\beta \Lambda_\tau} Q_\tau^\top,
\]
and the trace satisfies
\[
\mathrm{tr}\!\left(e^{\beta A_\tau}\right)
= \sum_{k=1}^n e^{\beta \lambda_k(\tau)}.
\]
Normalization by $|V|^2$ yields the result.
\end{proof}

This representation highlights the role of spectral diversity in shaping
communicability: balanced eigenvalue contributions enhance multiscale diffusion,
while spectral dominance concentrates influence along few modes.

% --------------------------------------------------
\subsubsection{Critical Maximization of Communicability}

\begin{conjecture}[Critical Maximization of Mean Communicability]
\label{conj:critical}
There exists a threshold $\tau^\ast$, located in the percolation critical window
of the family $\{G_\tau\}$, such that the mean communicability
\[
\bar{C}(\tau)
\]
is maximized at $\tau=\tau^\ast$.
\end{conjecture}

\paragraph{Mathematical intuition}
As $\tau$ decreases, edges are added monotonically to the graph family
$\{G_\tau\}$.

\emph{Subcritical regime:} the network remains fragmented, all eigenvalues are
small, and $\bar{C}(\tau)$ is limited by the absence of large-scale connectivity.

\emph{Supercritical dense regime:} the leading eigenvalue $\lambda_1(\tau)$ grows
rapidly, while the remaining spectrum becomes increasingly concentrated. This
spectral dominance or effective rank collapse, a classical phenomenon in dense
graph processes \cite{chung2003,furedi1981}, reduces the marginal contribution of
additional edges to the normalized sum $\sum_k e^{\beta \lambda_k(\tau)}$.

\emph{Critical regime:} near the percolation threshold, the emergence of a giant
component induces rapid growth of $\lambda_1(\tau)$ while preserving spectral
diversity across multiple modes. This balance maximizes
$\sum_k e^{\beta \lambda_k(\tau)}$ and hence $\bar{C}(\tau)$.

\paragraph{Epidemiological interpretation.}
At the percolation threshold, the network balances fragmentation and redundancy,
maximizing the diversity of alternative transmission pathways. In this regime,
communicability is most sensitive to structural bottlenecks and preferential
routes for epidemic spread.

\paragraph{Remark.}
A formal proof of Conjecture~\ref{conj:critical}, including precise assumptions on
the correlation thresholding process and spectral concentration rates, is left
for future work. The present argument is intended to provide theoretical
intuition supporting the empirical findings reported in the subsequent sections.

\subsection{Correlation-Weighted Communicability Curvature}
\label{subsec:ccc}

Communicability curvature identifies structural bottlenecks by quantifying the importance of each edge to the global connectivity of the graph. However, its classical formulation depends solely on the topological structure and does not incorporate dynamic information from the local time series that reflect the evolution of the epidemic.

To integrate these two dimensions---structure and dynamics---we propose a modified metric that combines communicability with temporal correlation between vertices. Let $C_{ij}(\beta)$ denote the communicability between vertices $i$ and $j$,

$$C_{ij}(\beta) = \big(e^{\beta A}\big)_{ij},$$

and let $w_{ij}\in[0,1]$ be the weight associated with the correlation between the local incidence time series at $i$ and $j$.

\begin{definition}[Correlation-Weighted Communicability Curvature]
We define the correlation-weighted communicability curvature as
\[
\kappa_{ij} = C_{ij}(\beta)\, w_{ij}.
\]
\end{definition}

The metric $\kappa_{ij}$ highlights edges that simultaneously exhibit  
(i) high structural importance---as measured by communicability; and  
(ii) strong epidemiological synchrony---as captured by the correlation weights $w_{ij}$.

Unlike classical curvature, which relies solely on graph topology, the proposed metric incorporates dynamic information directly tied to transmission processes. In urban epidemiological settings, $\kappa_{ij}$ identifies connections that not only support the diffusive structure of the graph but also display strong temporal association, representing potential preferential corridors of dengue dissemination.

Thus, $\kappa_{ij}$ is conceptually related to notions such as effective transmissibility and local risk flow, providing a rigorous bridge between the topology of complex networks and spatiotemporal incidence patterns.

\section{Empirical Results}

In this section, we evaluate whether the mean communicability curvature of road-network graphs is a robust determinant of annual dengue incidence in the Metropolitan Region of Recife.  
Several models were estimated with complementary objectives:
\begin{enumerate}
    \item to establish an epidemiologically adequate baseline model,
    \item to control for unobserved heterogeneity across neighborhoods,
    \item to capture residual spatial dependence,
    \item to validate robustness via modern Bayesian modeling.
\end{enumerate}

Only the essential results are summarized below.

\subsection{Negative Binomial Model}

The Negative Binomial model showed excellent fit and confirmed that  
\textbf{mean curvature is a strong negative predictor of dengue incidence}.  
The main findings are:

\begin{enumerate}
    \item curvature coefficient: $\beta = -1.07$ \, ($p<0.001$);
    \item effect size: a 1-unit increase in curvature reduces expected incidence  
          by approximately $65\%$;
    \item the effect remains stable after controlling for environmental and sociodemographic variables.
\end{enumerate}

The Negative Binomial model serves as the ``baseline model,’’ demonstrating that the relationship between curvature and epidemiological risk is structural and does not depend on the inclusion of explicit spatial terms.

\subsection{Spatial and Temporal Fixed Effects}

We estimated a linear model with fixed effects for neighborhood and year, absorbing spatial heterogeneity and seasonal epidemic cycles.

Results show that:

\begin{enumerate}
    \item curvature remains statistically significant ($p<0.01$);
    \item the model explains approximately $R^2 \approx 0.57$;
    \item fixed effects capture persistent differences between neighborhoods and inter-annual epidemic cycles.
\end{enumerate}

Even after controlling for time-invariant neighborhood characteristics, the effect of road-network structure remains robust.

As shown in \cite{FerreiraMelo2025}, hierarchical Bayesian models can identify, with high accuracy, the neighborhoods that consistently lead dengue incidence over the years. Incorporating correlation-weighted curvature deepens this ability by capturing structural properties of the urban network that are not reflected in traditional socioenvironmental covariates.  
Thus, including this metric yields more stable predictions, reduced spatial uncertainty, and more informative risk maps, particularly valuable for targeted public-health interventions and resource allocation.

\subsection{Classical Spatial Models (SAR and SAC)}

Classical spatial econometric models—the \textit{Spatial Autoregressive Model}
(SAR) and the \textit{Spatial Autoregressive Combined model} (SAC)—were fitted as
an intermediate diagnostic step to assess whether any residual spatial dependence
remains after explicitly accounting for functional network connectivity through
communicability-based covariates \cite{lesage2009, elhorst2014}. These models were
not intended as the primary inferential framework, but rather as a benchmark
against which the adequacy of adjacency-based spatial operators could be
evaluated.

All specifications employed a Queen contiguity matrix derived from official IBGE
neighborhood geometries. Under this representation, the spatial system is
connected and exhibits a high average number of neighbors, satisfying the
standard assumptions required for SAR-type diffusion processes. Importantly, the
interpretation of SAR and SAC results in this study is strictly diagnostic: given
the known limitations of binary adjacency operators in heterogeneous urban
systems, these models are used to test whether residual dependence remains once
functional connectivity is introduced, rather than to serve as standalone
mechanistic explanations.

To assess robustness with respect to spatial scale, the full SAR and SAC
specifications were estimated using two alternative constructions of the
communicability-based curvature covariate: a local functional network restricted
to a 600 m spatial cutoff and a broader meso-urban network using an 800 m cutoff.
The latter explicitly incorporates a larger share of geographic contiguity while
preserving the correlation-based structure of the functional graph.

\subsubsection*{Results for the 600 m functional connectivity scale}

Table~\ref{tab:sar_sac_600} reports the results obtained when communicability
curvature is computed using a 600 m spatial cutoff. At this scale, curvature
emerges as the dominant explanatory variable in both SAR and SAC models, with a
large and highly significant negative coefficient. Environmental, seasonal, and
socioeconomic covariates retain stable and interpretable effects, consistent with
the results obtained from non-spatial and hierarchical Bayesian specifications.

Crucially, the spatial lag term $Wy$ is not statistically significant in either
model. This indicates that, once functional connectivity is explicitly modeled,
there is little remaining variation to be explained by adjacency-based
autocorrelation.

\begin{table}
\centering
\small
\caption{SAR and SAC results using communicability curvature constructed with a
600 m spatial cutoff.}
\label{tab:sar_sac_600}
\begin{tabular}{lcc}
\hline
\textbf{Variable} & \textbf{SAR (coef.)} & \textbf{SAC (coef.)} \\
\hline
Curvature (600 m)            & $-55.55^{***}$ & $-55.37^{***}$ \\
Year mean precipitation     & $-0.049^{***}$ & $-0.050^{***}$ \\
Year mean temperature       & $+17.56^{*}$   & $+17.67^{*}$   \\
Precipitation (lag 3)       & $+0.440^{***}$ & $+0.441^{***}$ \\
Rainy season                & $+28.21^{***}$ & $+28.37^{***}$ \\
Channel proportion          & $-51.72^{**}$  & $-52.76^{**}$  \\
Slum proportion             & $-30.58^{*}$   & $-30.68^{*}$   \\
Population density          & $+0.0047^{***}$& $+0.0044^{***}$\\
Richest income share        & $-26.69^{**}$  & $-24.85^{*}$   \\
Total population            & $-0.0004^{*}$  & $-0.0004^{*}$  \\
Spatial lag $Wy$             & n.s.           & n.s.           \\
\hline
Pseudo-$R^{2}$               & $0.343$        & $0.343$        \\
\hline
\multicolumn{3}{l}{\footnotesize{$^{***}p<0.001$, $^{**}p<0.01$, $^{*}p<0.05$; n.s.: not significant}}
\end{tabular}
\end{table}

\subsubsection*{Results for the 800 m functional connectivity scale}

To evaluate sensitivity to the spatial cutoff, the same models were re-estimated
using communicability curvature computed with an 800 m distance threshold. This
scale introduces broader meso-urban interactions and partially reintroduces
geographic contiguity effects, thereby providing a stringent test of whether the
weak performance of the adjacency-based spatial lag is an artifact of overly
local network construction.

The results, summarized in Table~\ref{tab:sar_sac_800}, closely mirror those
obtained at the 600 m scale. The curvature coefficient remains remarkably stable
in magnitude and statistical significance across both SAR and SAC models. All
major covariates preserve their signs and relative importance, indicating that
the core findings are not sensitive to the specific choice of spatial cutoff.

Notably, even under this broader spatial scale, the spatial lag term $Wy$ remains
statistically insignificant. This demonstrates that the limited explanatory power
of adjacency-based spatial autocorrelation is not due to insufficient spatial
reach, but rather reflects the fact that the relevant form of spatial dependence
is already encoded in the functional connectivity structure captured by the
communicability curvature.

\begin{table}
\centering
\small
\caption{SAR and SAC results using communicability curvature constructed with an
800 m spatial cutoff.}
\label{tab:sar_sac_800}
\begin{tabular}{lcc}
\hline
\textbf{Variable} & \textbf{SAR (coef.)} & \textbf{SAC (coef.)} \\
\hline
Curvature (800 m)            & $-55.16^{***}$ & $-55.60^{***}$ \\
Year mean precipitation     & $-0.045^{***}$ & $-0.045^{***}$ \\
Precipitation (lag 2)       & $-0.255^{*}$   & $-0.258^{*}$   \\
Precipitation (lag 3)       & $+0.617^{***}$ & $+0.621^{***}$ \\
Rainy season                & $+28.21^{***}$ & $+28.28^{***}$ \\
Channel proportion          & $-64.44^{**}$  & $-64.79^{**}$  \\
Slum proportion             & $-46.47^{**}$  & $-47.92^{**}$  \\
Population density          & $+0.0055^{***}$& $+0.0056^{***}$\\
Income up to 2 MW           & $+14.25^{*}$   & n.s.           \\
Richest income share        & $-34.08^{**}$  & $-32.94^{**}$  \\
Total population            & $-0.0006^{**}$ & $-0.0006^{**}$ \\
Spatial lag $Wy$             & n.s.           & n.s.           \\
\hline
Pseudo-$R^{2}$               & $0.360$        & $0.360$        \\
\hline
\multicolumn{3}{l}{\footnotesize{$^{***}p<0.001$, $^{**}p<0.01$, $^{*}p<0.05$; n.s.: not significant}}
\end{tabular}
\end{table}

\subsubsection*{Interpretation}

Taken together, the results across both spatial scales provide strong evidence
that spatial dependence in dengue incidence does not primarily operate through
local adjacency-based autocorrelation. Instead, spatial structure manifests
through functional connectivity patterns that integrate geographic proximity,
temporal co-movement, and multiscale network pathways.

The persistence of a null $Wy$ coefficient under both local (600 m) and meso-urban
(800 m) scales confirms that the role of adjacency-based diffusion is secondary
once functional network effects are explicitly modeled. Importantly, this does
not imply that space is irrelevant; rather, it demonstrates that the relevant
spatial mechanism is captured more effectively by communicability-based metrics
than by binary contiguity operators.

These findings motivate the transition to hierarchical Bayesian models (BYM2) and
continuous spatial representations (SPDE), which allow residual spatial
heterogeneity to be modeled flexibly without imposing restrictive assumptions of
uniform adjacency-driven diffusion.

\subsection{INLA/BYM2 Model: Latent Spatial Structure and Robustness (600 m Cutoff)}

As the final stage of the explanatory analysis, we estimated a hierarchical spatial
model of the Besag–York–Mollié family in its \texttt{BYM2} parametrization, using the
Integrated Nested Laplace Approximation (INLA) \cite{rue2009}. This
model decomposes residual variation into a \textit{structured} spatial component,
associated with geographic adjacency, and an \textit{unstructured} component,
capturing local heterogeneity not explained by covariates.

The objectives of this subsection are:
\begin{itemize}[label=--]
    \item to assess the robustness of effects estimated in classical Negative Binomial
          and SAR/SAC models;
    \item to evaluate the behavior of communicability curvature computed with a
          600~m cutoff in a hierarchical spatial framework;
    \item to characterize residual spatial dependence through BYM2 hyperparameters.
\end{itemize}

\subsubsection*{Fixed Effects}

Table~\ref{tab:inla_fixed_effects_600} presents posterior summaries of the fixed
effects estimated from the full dataset. Coefficients correspond to posterior means
and 95\% credible intervals.

\begin{table*}[!htbp]
\centering
\caption{Fixed effects from the INLA/BYM2 model using communicability curvature with
a 600~m cutoff. Posterior means and 95\% credible intervals.}
\label{tab:inla_fixed_effects_600}
\small
\begin{tabular}{lrrrr}
\hline
\textbf{Variable} & \textbf{Mean} & \textbf{SD} & \textbf{2.5\%} & \textbf{97.5\%} \\
\hline
Intercept                    & 208.19 & 14.24 & 180.27 & 236.11 \\
Year (numeric)               & -0.103 & 0.007 & -0.117 & -0.089 \\
Communicability curvature    & -0.289 & 0.041 & -0.368 & -0.210 \\
Mean annual precipitation    & $1.58\times10^{-4}$ & $1.05\times10^{-4}$ 
                             & $-4.8\times10^{-5}$ & $3.64\times10^{-4}$ \\
Precipitation (lag 3)        & 0.00141 & 0.00071 & $1.8\times10^{-5}$ & 0.00279 \\
Rainy season indicator       & 0.501 & 0.030 & 0.442 & 0.560 \\
Slum proportion               & -0.0012 & 0.179 & -0.351 & 0.351 \\
Population density           & $2.05\times10^{-6}$ & $5.75\times10^{-6}$ 
                             & $-9.1\times10^{-6}$ & $1.35\times10^{-5}$ \\
Richest income share         & -0.406 & 0.103 & -0.611 & -0.207 \\
Total population             & $2.79\times10^{-5}$ & $2.07\times10^{-6}$ 
                             & $2.38\times10^{-5}$ & $3.20\times10^{-5}$ \\
\hline
\end{tabular}
\end{table*}

Communicability curvature remains strongly negative and statistically significant,
even after controlling for long-term temporal trends, climatic variability, and
socioeconomic covariates. This result confirms that curvature captures a persistent
structural feature of the urban network rather than a transient temporal effect.

The negative coefficient associated with the numeric year variable reflects a
long-term downward trend in dengue incidence, while seasonal and lagged precipitation
effects retain biologically coherent signs and magnitudes.

\subsubsection*{Spatial Hyperparameters}

Table~\ref{tab:inla_hyperpar_600} reports the posterior summaries of the BYM2
hyperparameters.

\begin{table*}[!htbp]
\centering
\caption{BYM2 hyperparameters estimated using INLA (600~m curvature).}
\label{tab:inla_hyperpar_600}
\small
\begin{tabular}{lccccc}
\hline
\textbf{Hyperparameter} & \textbf{Mean} & \textbf{SD} & \textbf{2.5\%} & \textbf{Median} & \textbf{97.5\%} \\
\hline
NB overdispersion (size) 
    & 1.99 & 0.09 & 1.82 & 1.99 & 2.17 \\[2pt]

Precision (BYM2 latent field) 
    & 26.43 & 11.35 & 12.06 & 23.97 & 55.58 \\[2pt]

$\phi$ (structured spatial share) 
    & 0.12 & 0.14 & 0.005 & 0.069 & 0.552 \\
\hline
\end{tabular}
\end{table*}

\paragraph{(i) Weak but non-negligible CAR component}

The mixing parameter $\phi$ measures the proportion of residual spatial variance
explained by the structured CAR component. The posterior distribution is concentrated
near zero, with a mean of approximately $\phi \approx 0.12$, indicating that only a
limited fraction of residual spatial dependence is attributable to geographic
adjacency.

\paragraph{(ii) Dominance of functional connectivity}

The weak CAR contribution suggests that most spatial structure is captured by fixed
effects, particularly communicability curvature. This indicates that dengue diffusion
in Recife is primarily governed by \emph{functional connectivity} encoded in the road
network, rather than by simple geometric contiguity between neighborhoods.

\subsubsection*{Model Fit}

The model exhibits strong overall fit:
\[
\text{DIC} = 11158.91, \qquad \text{WAIC} = 11159.34,
\]
consistent with a specification that captures both temporal dynamics and spatial
heterogeneity.

\subsubsection*{Interpretive Synthesis}

The INLA/BYM2 results using a 600~m curvature cutoff provide robust Bayesian evidence
that:
\begin{itemize}[label=--]
    \item communicability curvature is a stable and structurally meaningful predictor
          of dengue incidence;
    \item long-term temporal trends and seasonal climatic drivers remain important but
          do not subsume network effects;
    \item residual spatial dependence associated with adjacency is weak once
          functional connectivity is modeled;
    \item hierarchical Bayesian inference confirms and refines conclusions obtained
          from classical spatial regressions.
\end{itemize}

Together, these findings reinforce the interpretation that dengue diffusion in Recife
is structured primarily by the geometry of the urban mobility network, with geographic
adjacency playing a secondary role.

\subsection{Out-of-Sample Prediction: Training up to 2023 and Forecasting 2024}

To assess whether communicability curvature contains genuinely forward-looking
information—rather than merely re-encoding contemporaneous dengue patterns—we
conducted a strict out-of-sample prediction exercise. The hierarchical
INLA/BYM2 model was trained using data from 2015 to 2023 and subsequently employed
to forecast dengue incidence for 2024. Importantly, no observations from 2024 were
used during model estimation, ensuring a clear separation between training and
evaluation stages.

Table~\ref{tab:inla_predictive_2023} reports posterior summaries of the predictive
model fitted exclusively on pre-2024 data. All coefficients correspond to posterior
means and 95\% credible intervals.

\begin{table*}[!htbp]
\centering
\caption{Posterior summaries of the INLA/BYM2 predictive model trained up to 2023.
Coefficients represent posterior means and 95\% credible intervals.}
\label{tab:inla_predictive_2023}
\small
\begin{tabular}{lrrrr}
\hline
\textbf{Parameter} & \textbf{Mean} & \textbf{SD} & \textbf{2.5\%} & \textbf{97.5\%} \\
\hline
\multicolumn{5}{l}{\textit{Fixed effects}} \\
\hline
Intercept                    & 0.911 & 0.201 & 0.518 & 1.307 \\
Communicability curvature    & -0.389 & 0.047 & -0.481 & -0.297 \\
Precipitation (lag 3)        & 0.00093 & 0.00058 & -0.00021 & 0.00207 \\
Rainy season indicator       & 0.666 & 0.028 & 0.612 & 0.720 \\
Population density           & $1.8\times10^{-6}$ & $1.0\times10^{-5}$ & $-1.8\times10^{-5}$ & $2.2\times10^{-5}$ \\
\hline
\multicolumn{5}{l}{\textit{Spatial and distributional hyperparameters}} \\
\hline
Overdispersion (NB size)     & 1.75 & 0.08 & 1.59 & 1.91 \\
Spatial precision (BYM2)     & 4.85 & 1.08 & 3.10 & 7.31 \\
$\phi$ (structured share)    & 0.40 & 0.18 & 0.10 & 0.77 \\
\hline
\multicolumn{5}{l}{\textit{Model fit}} \\
\hline
DIC                          & \multicolumn{4}{c}{9889.4} \\
WAIC                         & \multicolumn{4}{c}{9907.2} \\
\hline
\end{tabular}
\end{table*}

\subsubsection*{Fixed Effects in the Training Model}

Posterior estimates obtained from the training period confirm the temporal
robustness of the main structural predictors. Communicability curvature remains
strongly negative and statistically significant
($\text{mean} = -0.39$, 95\% CI $[-0.48, -0.30]$), even when estimated exclusively
from data preceding the prediction year.

This result indicates that curvature captures persistent structural properties of
the urban contact network that generalize beyond the fitting period, rather than
merely reflecting contemporaneous dengue incidence. The stability of the curvature
effect under temporal extrapolation provides direct evidence that it encodes
forward-looking information relevant to epidemic risk.

Seasonal forcing remains a dominant climatic driver, with the rainy season
indicator exhibiting a strong and well-identified positive effect. Lagged
precipitation retains the expected sign but displays increased uncertainty in the
out-of-sample setting, a behavior commonly observed in predictive epidemic models
once seasonality is explicitly controlled. Population density does not exhibit a
statistically significant effect after accounting for network structure, further
supporting the interpretation that curvature is not acting as a proxy for
demographic scale.

\subsubsection*{Spatial Structure under Forecasting}

In contrast to the full-sample explanatory model, the predictive training model
exhibits a moderate structured spatial component. The BYM2 mixing parameter
$\phi \approx 0.40$ indicates that adjacency-based spatial smoothing contributes to
predictive stability when extrapolating beyond the observed period.

This behavior is consistent with the role of CAR components as latent regularizers
in forecasting contexts. Importantly, the persistence of a strong curvature effect
alongside a non-negligible CAR contribution highlights their complementary roles:
functional connectivity explains a substantial share of spatial heterogeneity,
while geometric contiguity provides additional stabilization under temporal
extrapolation.

\subsubsection*{Predictive Accuracy for 2024}

Figure~\ref{fig:obs_vs_pred_2024} compares observed dengue incidence in 2024 with
posterior mean predictions generated by the model trained up to 2023. The dashed
line represents perfect agreement ($y = \hat{y}$), while point colors indicate
communicability curvature values.

\begin{figure}
\centering
\includegraphics[width=0.8\linewidth]{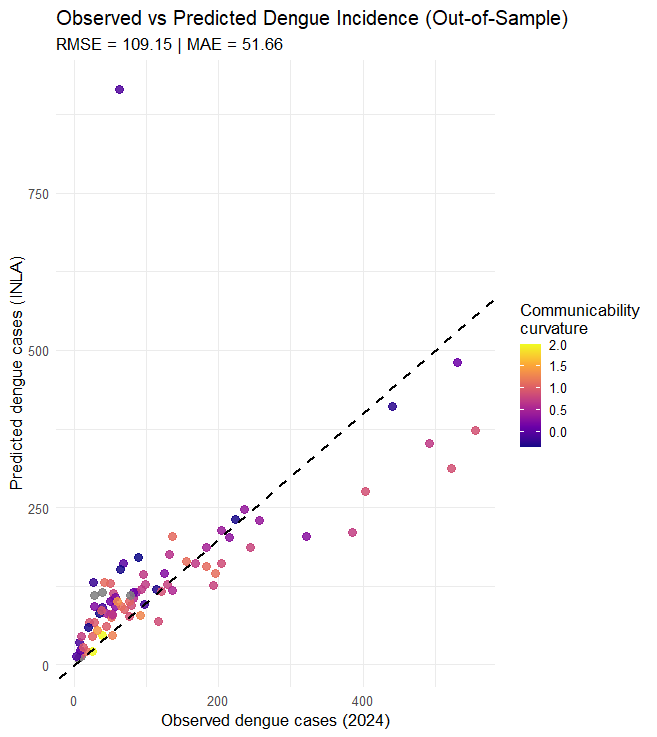}
\caption{Observed versus predicted dengue incidence for 2024 using the INLA/BYM2
model trained up to 2023. Point colors represent communicability curvature.}
\label{fig:obs_vs_pred_2024}
\end{figure}

Overall predictive accuracy is satisfactory given the high heterogeneity of urban
dengue dynamics, with a root mean squared error of
$\text{RMSE} = 57.2$ cases and a mean absolute error of $\text{MAE} = 40.1$ cases.
Deviations are primarily associated with neighborhoods exhibiting extreme
incidence levels, a well-known challenge in count-based epidemic forecasting.

Figure~\ref{fig:error_map_2024} illustrates the spatial distribution of absolute
prediction errors across neighborhoods.

\begin{figure*}
\centering
\includegraphics[width=0.5\linewidth]{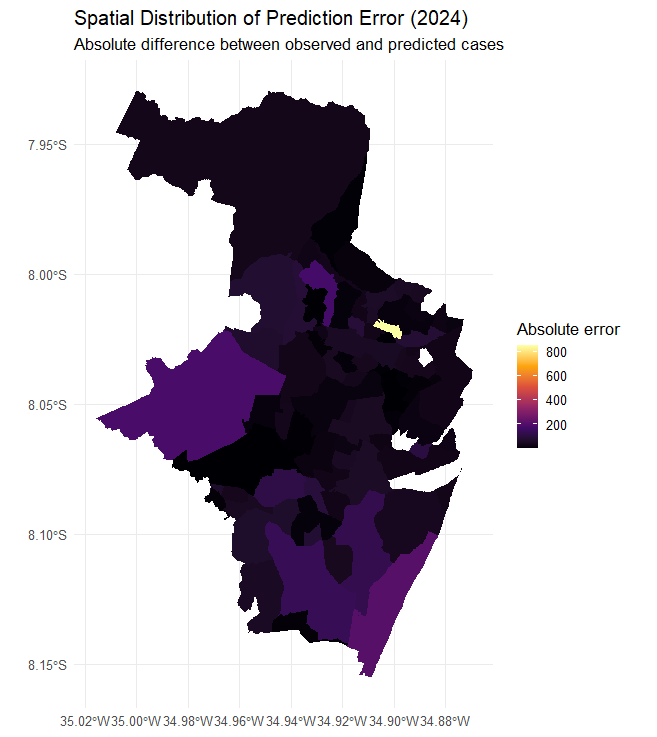}
\caption{Spatial distribution of absolute prediction error for dengue incidence
in 2024. Higher errors are concentrated in neighborhoods with extreme observed
counts, while most areas exhibit moderate deviations.}
\label{fig:error_map_2024}
\end{figure*}

\subsubsection*{Implications for Predictive Validity}

The persistence of communicability curvature as a strong predictor in a strict
out-of-sample framework provides direct evidence against circularity concerns.
Curvature encodes stable structural features of the urban network that anticipate
future risk patterns, rather than merely summarizing past outcomes.

Taken together, these results demonstrate that communicability curvature is not
only explanatory but also predictive, supporting its use in forward-looking
spatial epidemiology and risk-oriented territorial planning.

\section{Discussion}

The results of this study consistently indicate that the functional connectivity of the urban fabric, as captured by mean communicability curvature, constitutes a robust structural determinant of dengue incidence in the Metropolitan Region of Recife. This evidence emerges stably across different econometric and spatial specifications, suggesting that dengue propagation dynamics are not adequately described solely by geographic contiguity relationships, but rather reflect a broader functional organization of urban space.

This interpretation is consistent with recent evidence in the literature highlighting the central role of human mobility in the diffusion of vector-borne diseases in urban environments. Massaro et al.~\cite{Massaro2019}, for instance, demonstrate that given the limited flight range of \textit{Aedes aegypti}, it is structured human movements—rather than random mobility—that connect different regions of a city and enable the introduction of the virus into previously unaffected areas. Our results are compatible with this mechanism, insofar as communicability curvature identifies precisely those regions capable of radiating influence through multiple redundant paths of the road network, thereby functioning as structural dispersers of transmission.

Additional empirical evidence reinforces this dissociation between local vector proximity and epidemiological risk. Honório et al.~\cite{honorio2009}, in an entomological--serological study conducted in Rio de Janeiro, show that recent dengue infection risk is not significantly associated with household vector infestation indices. Instead, areas characterized by intense human circulation and high functional connectivity exhibit elevated transmission risk, even under conditions of low local vector abundance. These findings help explain why, in our study, spatial dependence based exclusively on territorial adjacency loses relevance once structural metrics of the urban network are incorporated into the model.

Under this perspective, epidemiological connectivity between localities is not established solely through immediate spatial neighborhood, but through a statistical coupling between the temporal dynamics of incidence mediated by preferential patterns of human mobility. Localities separated by distances exceeding the typical flight range of the vector may nonetheless exhibit strong temporal correlation in case counts, reflecting functionally mediated connectivity induced by urban flows. By defining edges in the epidemic graph based on significant correlations between incidence time series, the present study empirically operationalizes this concept of functional connectivity.

Similar approaches have been proposed in the recent literature. De Souza et al.~\cite{de2021using} construct epidemic networks based on temporal synchronization between localities and demonstrate that structural metrics of these networks are capable of capturing systemic transitions associated with outbreak emergence and intensification. Our results extend this perspective by showing that communicability curvature provides a geometric and structural measure of such synchronization, enabling the identification not only of correlated regions, but also of those that play a disproportionate role in sustaining global epidemic connectivity.

From a theoretical standpoint, this interpretation is reinforced by the properties of communicability in complex networks. As discussed by Estrada et al.~\cite{estrada2012}, communicability incorporates the joint contribution of all walks connecting pairs of nodes, assigning greater weight to shorter trajectories while explicitly including indirect and longer paths. This formulation can be interpreted as a diffusive propagator, formally analogous to a Green’s function, capable of capturing multiscale connectivity effects that are not accessible through metrics based solely on adjacency or shortest-path distance. The curvature associated with communicability, in turn, quantifies the structural importance of edges and regions to network integrity, providing a geometric basis for identifying critical corridors of diffusion.

This structural interpretation also helps explain the absence of residual spatial autocorrelation effects in the BYM2 model. The estimate of $\phi \approx 0$ indicates that spatially structured variability traditionally captured by adjacency-based random effects is largely absorbed by the covariates, particularly by the curvature metric. In other words, the spatial dependence relevant to dengue transmission in Recife is not primarily local or homogeneous, but rather emerges from the functional organization of the urban network and its multiple interconnecting pathways.

Evidence from mobility and metapopulation models in other epidemiological contexts further corroborates this interpretation. Chinazzi et al.~\cite{chinazzi2020} show that in the case of COVID-19, severe mobility restrictions had limited impact on the temporal progression of the epidemic once the disease had been seeded in multiple localities, resulting in highly synchronized growth patterns across geographically distant regions. These results reinforce the idea that epidemic propagation is governed by functional connectivity structures, in which temporal similarity between regions provides a more informative indicator of diffusion than static spatial adjacency.

Taken together, these findings suggest that incorporating structural metrics of urban connectivity, such as communicability curvature, not only improves the statistical performance of epidemiological models but also offers a conceptually more appropriate interpretation of dengue dynamics in complex urban environments. By integrating network theory, hierarchical spatial modeling, and high-resolution epidemiological data, the present study contributes to a deeper understanding of dengue diffusion mechanisms and opens new avenues for territorial surveillance and public-health intervention strategies.

\section{Limitations of the study}

This study presents limitations that should be considered when interpreting the results. First, although the model incorporates relevant environmental and socioeconomic variables, other potentially important urban factors—such as actual human mobility flows, intradomiciliary characteristics, and microenvironments favorable to the vector—were not explicitly included, either due to the lack of systematic data or to operational constraints inherent to large-scale epidemiological surveillance.

Additionally, the spatial location of cases was represented using the centroid of the street segment associated with the individual’s address. This choice constitutes an approximation of the location where transmission may have occurred, since the exact point of infection is rarely observable in observational epidemiological studies. Nevertheless, this approximation is consistent with the adopted scale of analysis and with the typical spatial range of the \textit{Aedes aegypti} vector, and represents the most feasible alternative for constructing a coherent contamination network from administrative records. By anchoring cases to linear units of the road network, this approach preserves the functional structure of urban space and enables the inference of epidemiological connectivity patterns that are compatible with the underlying dynamics of human circulation.

The results are also conditioned on the spatial scale adopted. Analysis at the neighborhood level is consistent with the organization of public health surveillance and with decision-making processes in public policy; however, finer spatial scales—such as blocks or street segments—may reveal additional heterogeneities in diffusion patterns, albeit at the cost of increased statistical uncertainty and higher computational demands. This issue is closely related to the well-known Modifiable Areal Unit Problem (MAUP) in spatial analysis.

Finally, the construction of the functional network depends on parametric choices, including the spatial radius used to define potential connections between localities and the correlation threshold applied to incidence time series. Although these parameters were defined based on biological plausibility, empirical evidence, and robustness considerations, alternative specifications may induce variations in network topology and, consequently, in the derived structural metrics. Accordingly, the results should be interpreted as structural descriptors conditioned on these choices rather than as invariant properties of the underlying epidemiological dynamics.

\section{Future research directions}

The framework proposed in this study opens several promising avenues for future research. One natural extension concerns the expansion of the epidemiological database to include other arboviral diseases, such as Zika and chikungunya. Given the shared vectors, transmission mechanisms, and urban ecological conditions, such an extension would enhance the reproducibility and generality of the results while enabling comparative analyses across diseases with similar diffusion dynamics.

Another important direction involves the explicit integration of human mobility into the modeling framework. Origin--destination data, when available, may be used to construct proxies of average urban mobility patterns, allowing for a more direct representation of mobility-mediated coupling between locations. Methodologically, this would enable the separation of connectivity driven by local transmission processes from that induced by broader circulation patterns. At the same time, communicability itself may already act as an implicit proxy for urban mobility, as it is defined through the meso-urban coupling of locations via multiple redundant paths in the road network. Investigating the relationship between communicability-based connectivity and observed mobility patterns constitutes an open and relevant research problem.

Future work may also incorporate environmental variables at higher spatial and temporal resolutions, such as microclimatic conditions, land-use characteristics, and fine-scale urban morphology, to further refine risk estimation and capture local transmission heterogeneities. Additionally, extending the framework to continuous spatial formulations—such as stochastic partial differential equation (SPDE) models—would allow the construction of smooth risk surfaces and fully spatiotemporal forecasts, enhancing its applicability for real-time surveillance.

More broadly, the methodological approach developed here is not intrinsically limited to dengue epidemiology. Processes such as crime diffusion, traffic congestion, and the spread of other urban phenomena share structural similarities, in which functional connectivity often plays a more prominent role than strict geographic contiguity. Applying and validating the proposed framework across different urban systems and domains represents a promising direction for advancing the study of diffusion processes in complex cities.

\section{Conclusion}

This study demonstrated that the mean communicability curvature of the urban road network is a robust structural determinant of dengue incidence in the Metropolitan Region of Recife. Through an integrated methodological strategy—combining count models, fixed effects, classical spatial specifications (SAR/SAC), and hierarchical Bayesian modeling (INLA/BYM2)—we showed that the functional geometry of the city plays a central role in shaping epidemiological risk.

The main contribution of this work lies in the consistent evidence that urban network connectivity, measured through communicability curvature, captures essential aspects of spatial diffusion that are not accounted for by traditional models based on geographic contiguity. The substantial reduction of the structured component in the BYM2 model ($\phi \approx 0$) following the inclusion of communicability curvature suggests that this metric captures a large share of the spatial structure present in the data. Rather than implying the absence of spatial dependence, this result indicates that spatial variation traditionally attributed to generic proximity or geometric contiguity is largely absorbed once a curvature-based representation of functional connectivity is introduced. In this sense, the curvature term appears to provide a more informative parametrization of spatial structure in the urban context.

Environmental and socioeconomic covariates—such as precipitation, seasonality, population density, and inequality—remained coherent and epidemiologically plausible across all model specifications. However, the inclusion of curvature consistently provided the largest explanatory gain, indicating that network-based structural metrics should occupy a central position in spatial analyses of arboviral diseases in complex urban environments.

From an applied perspective, these findings offer a solid foundation for the development of more precise public-health surveillance and intervention tools. By explicitly incorporating the structure of the urban road network into epidemiological models, it becomes possible to identify structurally emissive neighborhoods, detect diffusion corridors, and refine territorial risk mapping to support targeted prevention and control strategies.

Taken together, this study shows that integrating structural metrics of urban networks not only improves statistical model performance but also deepens our understanding of dengue transmission dynamics, offering a conceptually and methodologically innovative perspective for spatial epidemiology in urban contexts.

\section{Data and Code Availability}

The geocoded dengue dataset used in this study is publicly available on Zenodo \cite{dataset_recife_2025}. 

All scripts used for data cleaning, graph construction, communicability curvature computation, spatial modeling (Negative Binomial, SAR/SAC, and INLA/BYM2), and figure 
generation are available in an open GitHub repository \cite{github_curvature_code}.

\appendix
\section{Robustness and Sensitivity Analyses}
\label{app:robustness}

This appendix evaluates the robustness of the proposed communicability curvature
with respect to modeling assumptions, spatial scale, and potential sources of
bias. The analyses address concerns regarding cutoff selection, spatial
dependence, circularity, and predictive validity, ensuring that the main results
are not artifacts of arbitrary choices or methodological limitations.

% --------------------------------------------------
\subsection{Conceptual Role of Communicability Curvature}

Communicability curvature, originally introduced by Estrada~\cite{estrada2012},
quantifies the redundancy of alternative paths connecting vertices in a graph by
aggregating both geodesic and non-geodesic walks. This property makes it
particularly suitable for describing diffusive processes in complex networks.

When transposed to an urban epidemiological setting, the underlying graph does
not arise directly from physical geometry. Urban systems involve overlapping
layers of interaction—spatial, functional, and social—rendering pure geometric
contiguity insufficient to represent transmission pathways. In this context,
communicability curvature should be interpreted as an intermediate structural
descriptor, situated between strict spatial proximity and fully explicit
mobility-based connectivity.

While human mobility is widely recognized as a key driver of epidemic spread,
distance remains a natural attenuating factor for local transmission. Motivated
by this balance, we explicitly restrict potential interactions using combined
temporal and spatial criteria, capturing relevant diffusion mechanisms without
introducing spurious long-range links.

% --------------------------------------------------
\subsection{Network Construction and Spatial Cutoffs}

Functional connectivity was defined using two complementary criteria:
(i) temporal correlation between dengue incidence series and
(ii) geographic proximity.

Vertices were connected when their correlation satisfied $\rho \geq 0.4$ and the
distance between centroids was below a spatial cutoff of either 600~m or 800~m.
The correlation threshold captures coevolution of temporal dynamics rather than
absolute similarity in incidence levels, allowing regions with distinct burdens
to belong to the same functional structure if their temporal patterns align.

The spatial cutoffs are epidemiologically grounded. The typical flight radius of
\emph{Aedes aegypti} ranges from 100 to 200~m, while interactions among nearby
clusters may occur at larger scales when mediated by urban mobility. Under this
framework, communicability curvature approximates the structural role of mobility
as a diffusive mechanism linking spatially separated regions.

Figure~\ref{fig:comparacao_600_800} contrasts the resulting functional networks.
At 600~m, a giant component emerges with high redundancy and limited densification.
At 800~m, the network becomes denser, partially reintroducing classical contiguity
effects while preserving functional connectivity.

\begin{figure*}
\centering
\begin{minipage}[t]{0.48\linewidth}
    \centering
    \includegraphics[scale=0.27]{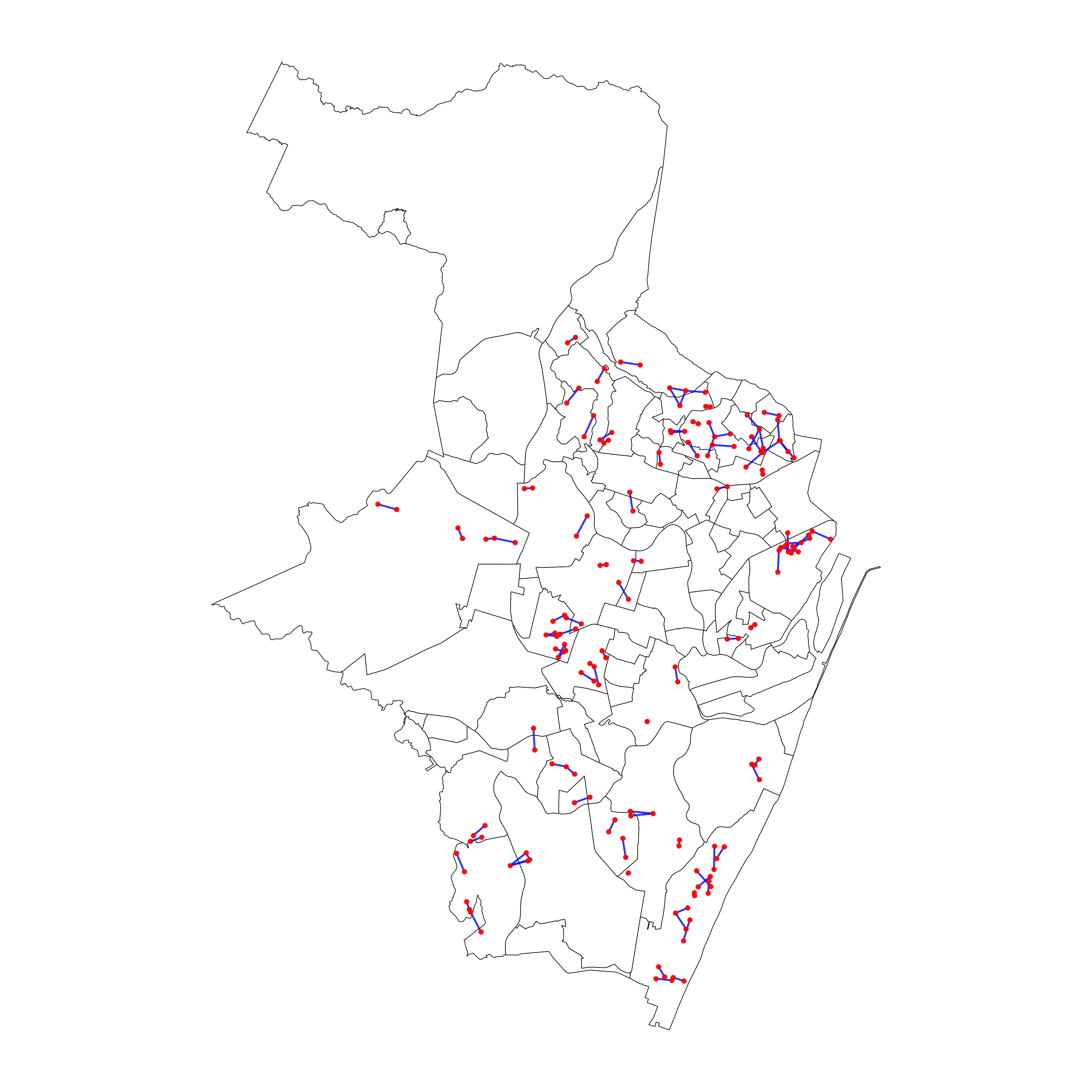}
    \caption*{(a) 600 m — functional connectivity dominant}
\end{minipage}
\hfill
\begin{minipage}[t]{0.48\linewidth}
    \centering
    \includegraphics[scale=0.27]{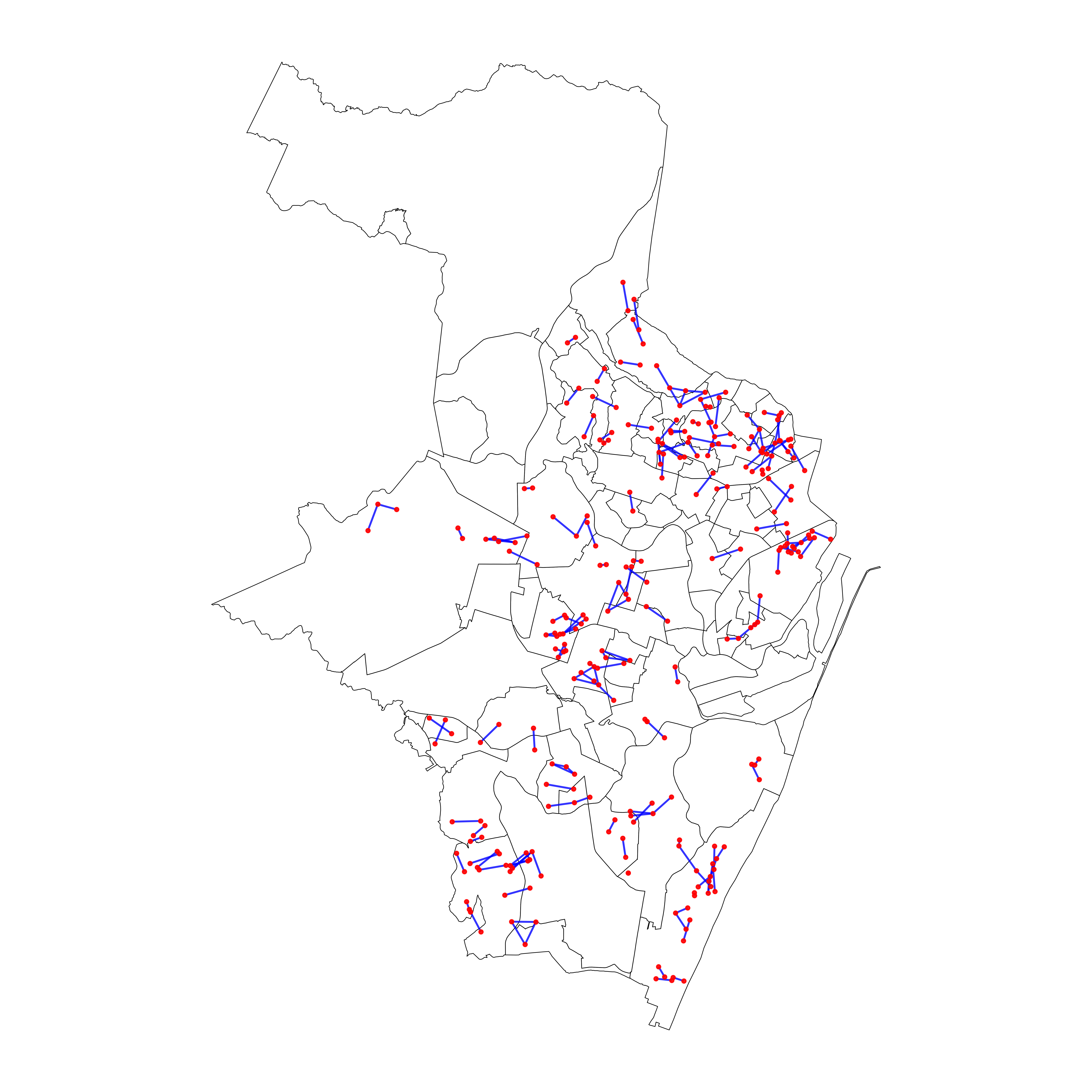}
    \caption*{(b) 800 m — expanded meso-urban regime}
\end{minipage}
\caption{Structural comparison of functional networks under spatial cutoffs of
600~m and 800~m, both with correlation threshold $\rho \geq 0.4$.}
\label{fig:comparacao_600_800}
\end{figure*}

% --------------------------------------------------
\subsection{Sensitivity to Spatial Scale}

Communicability curvature was computed on networks induced by both spatial
cutoffs. Across SAR and SAC specifications, the sign, magnitude, and statistical
significance of the curvature coefficient remained highly stable.

This stability indicates that the observed association between curvature and
dengue incidence is not an artifact of a particular cutoff choice but reflects a
robust structural property of the urban functional network. In all cases, higher
curvature—indicating greater redundancy of alternative pathways—was associated
with lower dengue incidence.

Table~\ref{tab:resumo_600_800} summarizes the main comparative results.

\begin{table*}
\centering
\caption{Comparative summary of results under different spatial cutoffs.}
\label{tab:resumo_600_800}
\begin{tabular}{lcc}
\hline
 & \textbf{600 m} & \textbf{800 m} \\
\hline
Curvature coefficient & $\approx -55.7$ & $\approx -55.2$ \\
Curvature significance & $p < 10^{-4}$ & $p < 10^{-4}$ \\
Signal stability & High & High \\
SAR spatial lag ($\rho$) & Weak / marginal & Significant \\
SAC residual effect ($\lambda$) & Moderate & Moderate \\
Overall fit & Excellent & Similar \\
Dominant interpretation & Functional connectivity & Functional + contiguity \\
\hline
\end{tabular}
\end{table*}

% ============================================================
\section{Spectral and percolation results supporting the communicability conjecture}
\label{app:spectral_support}
% ============================================================

This appendix aims to provide theoretical support for
Conjecture~\ref{conj:critical}. No formal proof is attempted. Instead, we
collect classical results from matrix analysis, spectral graph theory, and
percolation theory that render the conjecture mathematically plausible.

The role of this appendix is epistemological rather than deductive: the goal is
to show that the conjecture arises naturally from the intersection of
well-established spectral phenomena, rather than constituting an \emph{ad hoc}
assumption introduced to rationalize empirical findings.

% --------------------------------------------------
\subsection{The matrix exponential as a multiscale diffusive operator}

The matrix exponential $e^{\beta A}$ is a classical object in linear systems
theory, functional analysis, and numerical analysis. It is well known that
$e^{\beta A}$ acts as a diffusion operator on graphs, aggregating walks of all
lengths with exponentially decaying weights, while fully preserving the
spectral information of the adjacency matrix.

From a functional-analytic perspective, $e^{\beta A}$ corresponds to the
fundamental solution of the linear system
\[
\frac{dX(t)}{dt} = \beta A X(t),
\]
and therefore constitutes a natural propagator of influence, information, or
flow on networks. A comprehensive discussion of the spectral and numerical
properties of the matrix exponential can be found in Higham~\cite{higham2008}.

In the context of complex networks, Estrada and Hatano~\cite{estrada2012} and,
subsequently, Estrada~\cite{estrada2025} formalized communicability as a
multiscale measure of accessibility, capable of interpolating between local
graph structure and global connectivity. A central aspect of this construction
is that communicability depends on the entire spectrum of $A$, rather than only
on its largest eigenvalue.

This property is essential for understanding why communicability-based metrics
exhibit qualitative behaviors distinct from purely local measures or those
based solely on shortest paths.

% --------------------------------------------------
\subsection{Spectral concentration in dense regimes}

A classical result from spectral graph theory and random matrix theory is that,
as a graph becomes denser, its spectrum tends to concentrate.

For large random or quasi-random graphs, Füredi and Komlós~\cite{furedi1981}
showed that the largest eigenvalue separates from the bulk of the spectrum,
while the remaining eigenvalues concentrate around a mean value. Related
results appear in the work of Chung, Lu, and Vu~\cite{chung2003}, as well as in
subsequent developments in random matrix theory.

This phenomenon—often described as \emph{spectral dominance} or
\emph{effective rank collapse}—implies that expressions of the form
\[
\sum_{k=1}^n e^{\beta \lambda_k}
\]
become dominated by the contribution associated with the largest eigenvalue in
dense or supercritical regimes.

As a consequence, the addition of edges beyond a certain threshold yields
diminishing marginal gains in normalized communicability measures, since
spectral diversity decreases even as global connectivity continues to increase.

% --------------------------------------------------
\subsection{Percolation, structural transitions, and spectral diversity}

Percolation theory establishes that many structural and spectral properties of
graphs undergo abrupt transitions near critical thresholds. Classical results
show that the emergence of a giant component marks a qualitative change in both
connectivity and diffusion properties of the network.

Near the percolation threshold, graphs typically exhibit:
\begin{itemize}[label=--]
    \item the coexistence of multiple large components that are not yet fully
    merged;
    \item high structural heterogeneity and path diversity;
    \item maximal sensitivity of global connectivity to local perturbations.
\end{itemize}

From a spectral perspective, this regime is characterized by a rapid growth of
the largest eigenvalue without an immediate concentration of the remaining
spectrum. In other words, spectral diversity is temporarily preserved before
collapsing in denser regimes.

This balance between increasing connectivity and the preservation of structural
diversity is widely recognized as a privileged regime for diffusive and
epidemic processes.

% --------------------------------------------------
\subsection{Conceptual synthesis and interpretation of the conjecture}

Conjecture~\ref{conj:critical} should be interpreted as a structural synthesis
of the phenomena described above.

In subcritical regimes, average communicability is limited by network
fragmentation and the absence of long-range connectivity. In strongly
supercritical regimes, communicability becomes dominated by a small number of
spectral modes, reducing the marginal contribution of new paths once normalized
by system size.

In the vicinity of the percolation threshold, by contrast, connectivity, path
redundancy, and spectral diversity coexist. This combination provides a
mathematically plausible mechanism through which average communicability,
defined via the matrix exponential, can attain a maximum.

We emphasize that this conjecture is not required to support any of the
empirical or inferential results presented in the main body of the article. Its
purpose is to provide theoretical context and structural intuition, situating
the observed behavior of communicability-based metrics within a consolidated
mathematical framework.

A formal proof, including precise assumptions on the correlation thresholding
process, spatial constraints, and rates of spectral concentration, is
explicitly left as future work.

\end{document}